\documentclass[12pt]{article}

\usepackage{graphicx}
\usepackage{amsmath,amsthm,latexsym,amssymb,amsfonts,epsfig,MnSymbol,proof}

\newcommand{\1}{\mathbf{1}}

\newcommand{\emi}{\mathbb{E}}
\newcommand{\ntt}{N}
\newcommand{\ppx}{X}
\newcommand{\pra}{X}

\newcommand{\qry}{\mathbf{Q}}
\newcommand{\qs}{Q}
\newcommand{\rlv}{R}
\newcommand{\samps}{\Omega}

\newcommand{\nrlv}{\overline{\rlv}}
\newcommand{\nra}{\ntt_{\pra \rlv}}
\newcommand{\nrr}{\ntt_{\rlv}}

\newcommand{\prx}{x}

\newcommand{\nnr}{\ntt_{\nrlv}}
\newcommand{\nna}{\ntt_{\pra \nrlv}}

\newtheorem{theorem}{Theorem}
\newtheorem{lemma}[theorem]{Lemma}

\DeclareMathOperator{\trc}{Tr}

\title{Quantum emulation of query extension in information retrieval}

\author{
Rom\`an R. Zapatrin\\
\small\itshape
Informatics Dept., The State Russian Museum,\\
\small\itshape
In\.zenernaya 4, 191186, St.Petersburg, Russia;\\
\small\rm e-mail: Roman.Zapatrin at Gmail.com }

\begin{document}

\maketitle

\begin{abstract}
An operationalistic scheme, called Melucci metaphor, is suggested representing Information Retrieval as physical measurements with beam of particles playing the role of the flow of retrieved documents. The possibilities of query expansion by extra term are studied from this perspective, when the particles-`docuscles' are assumed to be of classical or quantum nature. It is shown that in both cases the choice of an extra term based on Bayesian belief revision is still valid on the qualitative level for boosting the relevance of the retrieved documents. 
\end{abstract}

\section{Basic settings: Melucci metaphor}\label{sbb}

Nowadays Information Retrieval, which passed a long way from data retrieval towards knowledge mining deals with new kind of reality, to say more, the new kind of reality is produced. The amount of data became immense, furthermore, their nature changed. Many documents are no longer static entities, rather, their emerge during the act of retrieval. As a result, new kind of theory is demanded, in which the r\^ole of a particular document is miserated. An analogy of this is thermodynamics, where the characteristics of a particular molecule plays no r\^ole, and the observed and employed results are of cumulative nature. 

Within the suggested approach, two parties are considered to be involved in the IR process: the User and the Engine. The User submit queries, they somehow affect the Engine, as a result she produces a sequence of replies. Each reply is evaluated by the User to be relevant or not. One more thing is that the notion of relevance becomes more and more ambiguous. I by no means try to formalize it; the main assumption is that after each query act the User decides (we need not specify, how it happens) which results are relevant. So, relevance is a verifiable property. Besides the relevance, one can check the occurrence of this or that term in each document and again, the occurrence of a particular term is a verifiable property of the document. As a consequence, we see that the Information Retrieval process is similar to measuring physical quantities: the User affects the Engine, the Engine returns results, the User verifies their properties. We consider a simplified model of IR process, called Melucci metaphor.

\paragraph{Melucci metaphor. Docuscles.} The documents are treated as particles, emitted from a source, call it the Emitter $\emi$; we call these particles `docuscles’, they are emitted in certain state $\qs$, which, in turn, depends on the User's query $\qry$. The docuscles arrive to a detector –- the User observes the results, decides if they are relevant, furthermore, the User can test other properties. The main property to test is $\rlv$, the relevance (with respect to query $\qry$), so each docuscle passing through the Testing Appliance produces a click signaling its relevance either non-relevance. The number of emitted docuscles is assumed to be potentially infinite, but the setup of the experiment is such that only first $\ntt$ docuscles are tested to be relevant. 
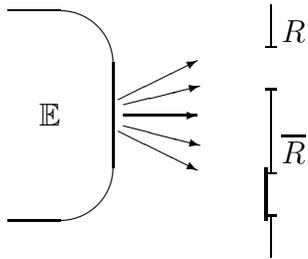
\unitlength.7mm
\[
\begin{picture}(180,60)
\put(0,30){\oval(40,40)[r]}
\put(6,28){$\emi$}
\put(21,33){\vector(2,1){15}}
\put(22,32){\vector(4,1){14.5}}
\put(22,30){\vector(1,0){14}}
\put(22,28){\vector(4,-1){14.5}}
\put(21,27){\vector(2,-1){15}}
\put(50,3){\line(0,1){8}}
\put(50,19){\line(0,1){16}}
\put(50,43){\line(0,1){8}}
\multiput(49,11)(0,8){2}{\line(1,0){2}}
\put(52,21){$\nrlv$}
\multiput(49,35)(0,8){2}{\line(1,0){2}}
\put(52,44){$\rlv$}
\multiput(48.8,10)(0.1,0){4}{\line(0,1){10}}
\end{picture}
\]
The number of relevant (irrelevant) documents is denoted by $\nrr$ ($\nnr$, respectively). So, $\ntt$ decomposes into
\[
\ntt=\nrr+\nnr
\]
That is, we have the frequency, which is interpreted as probability
\[
r=P({\rlv})=\frac{\nrr}{\ntt}\quad ; \qquad P({\nrlv})=\frac{\nnr}{\ntt}=1-r
\]

\paragraph{Query expansion.} After the relevance test is carried out, each beam is further examined by the next device, which tests the presence of certain property $\pra$. The number of relevant (non-relevant) docuscles having the property $\pra$ is denoted by $\nra$ ($\nna$, respectively). In terms of frequencies definition of probability we possess the following measured values:
\[
p=P(\pra\vert {\rlv})=\frac{\nra}{\nrr}\quad ; \qquad q=P(\pra\vert {\nrlv})=\frac{\nna}{\nnr}
\]
From this experiment we get the value of $P(\ppx\vert\rlv)$:
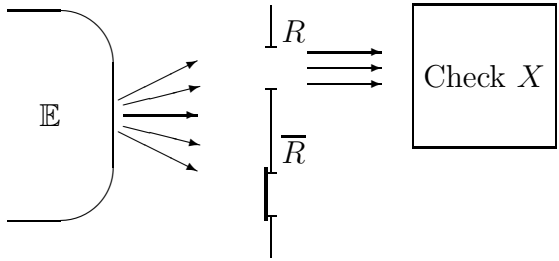
\[
\begin{picture}(180,60)
\put(0,30){\oval(40,40)[r]}
\put(6,28){$\emi$}
\put(21,33){\vector(2,1){15}}
\put(22,32){\vector(4,1){14.5}}
\put(22,30){\vector(1,0){14}}
\put(22,28){\vector(4,-1){14.5}}
\put(21,27){\vector(2,-1){15}}
\put(50,3){\line(0,1){8}}
\put(50,19){\line(0,1){16}}
\put(50,43){\line(0,1){8}}
\multiput(49,11)(0,8){2}{\line(1,0){2}}
\put(52,21){$\nrlv$}
\multiput(49,35)(0,8){2}{\line(1,0){2}}
\put(52,44){$\rlv$}
\multiput(48.8,10)(0.1,0){4}{\line(0,1){10}}
\multiput(57,36)(0,3){3}{\vector(1,0){14}}
\put(77,24){\framebox(27,27){Check $\ppx$}}
\end{picture}
\]
and from this experiment we get the value of $P(\ppx\vert\nrlv)$:
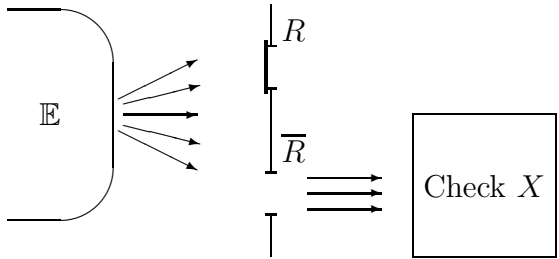
\[
\begin{picture}(180,55)
\put(0,30){\oval(40,40)[r]}
\put(6,28){$\emi$}
\put(21,33){\vector(2,1){15}}
\put(22,32){\vector(4,1){14.5}}
\put(22,30){\vector(1,0){14}}
\put(22,28){\vector(4,-1){14.5}}
\put(21,27){\vector(2,-1){15}}
\put(50,3){\line(0,1){8}}
\put(50,19){\line(0,1){16}}
\put(50,43){\line(0,1){8}}
\multiput(49,11)(0,8){2}{\line(1,0){2}}
\put(52,21){$\nrlv$}
\multiput(49,35)(0,8){2}{\line(1,0){2}}
\put(52,44){$\rlv$}
\multiput(48.8,34)(0.1,0){4}{\line(0,1){10}}
\multiput(57,12)(0,3){3}{\vector(1,0){14}}
\put(77,3){\framebox(27,27){Check $\ppx$}}
\end{picture}
\]
When we are in the classical realm, there is no need to calculate $P(\ppx)$ due to our Boolean belief revision (that is, the law of total probability):
\begin{equation}\label{eltp}
P(\ppx)=P(\ppx\vert\rlv)\,P(\rlv)+P(\ppx\vert\nrlv)\,P(\nrlv)
\end{equation}
But if we decide to measure $P(\ppx)$ directly, removing the relevance check:
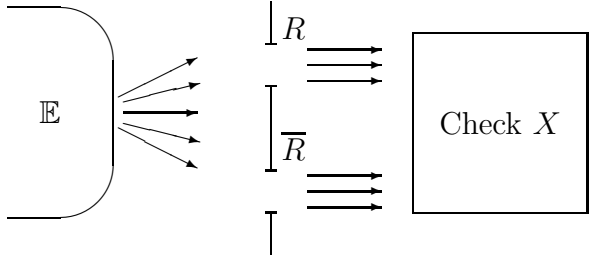
\[
\begin{picture}(180,55)
\put(0,30){\oval(40,40)[r]}
\put(6,28){$\emi$}
\put(21,33){\vector(2,1){15}}
\put(22,32){\vector(4,1){14.5}}
\put(22,30){\vector(1,0){14}}
\put(22,28){\vector(4,-1){14.5}}
\put(21,27){\vector(2,-1){15}}
\put(50,3){\line(0,1){8}}
\put(50,19){\line(0,1){16}}
\put(50,43){\line(0,1){8}}
\multiput(49,11)(0,8){2}{\line(1,0){2}}
\put(52,21){$\nrlv$}
\multiput(49,35)(0,8){2}{\line(1,0){2}}
\put(52,44){$\rlv$}
\multiput(57,36)(0,3){3}{\vector(1,0){14}}
\multiput(57,12)(0,3){3}{\vector(1,0){14}}
\put(77,11){\framebox(33,34){Check $\ppx$}}
\end{picture}
\]
the result may in general violate the Law of Total Probability \eqref{eltp} if the particles we are dealing with are of quantum nature. Why Quantum Mechanics becomes popular for IR models such as, say \cite{Piwowarski-Rijsbergen}? The first reason is that these models produce results. The second reason is that from the operationalistic point of view it resembles modern IR. Properties of a corpuscle do not pre-exist before a measurement act is carried out, rather, they emerge during the measurement act. 

\paragraph{Relevance boost.} Query expansion makes sense if we gain something, namely, we would like to increase the rate of relevant documents among the retrieved ones. For that, the search setting is modified. After the beam of docuscles is emitted in state $\qs$, it is first subject for selection to possess the property $\pra$. Then the relevance check for $\ntt$ (post-selected!) docuscles is carried out and we are left with the ratio
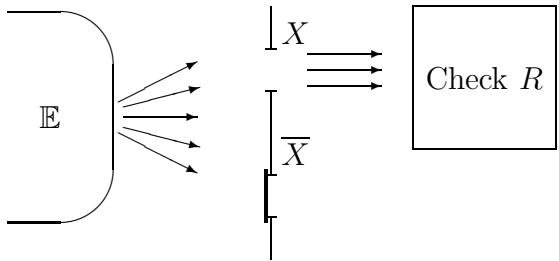
\[
\prx=
\frac{\rlv_\pra}{\ntt}
\]
\[
\begin{picture}(180,60)
\put(0,30){\oval(40,40)[r]}
\put(6,28){$\emi$}
\put(21,33){\vector(2,1){15}}
\put(22,32){\vector(4,1){14.5}}
\put(22,30){\vector(1,0){14}}
\put(22,28){\vector(4,-1){14.5}}
\put(21,27){\vector(2,-1){15}}
\put(50,3){\line(0,1){8}}
\put(50,19){\line(0,1){16}}
\put(50,43){\line(0,1){8}}
\multiput(49,11)(0,8){2}{\line(1,0){2}}
\put(52,21){$\overline{\ppx}$}
\multiput(49,35)(0,8){2}{\line(1,0){2}}
\put(52,44){$\ppx$}
\multiput(48.8,10)(0.1,0){4}{\line(0,1){10}}
\multiput(57,36)(0,3){3}{\vector(1,0){14}}
\put(77,24){\framebox(27,27){Check $\rlv$}}
\end{picture}
\]
We interpret this the result of query expansion as the conditional probability
\begin{equation}\label{edefx}
\prx=P({\rlv}\vert \pra)
\end{equation}
Relevance boost is a numerical characteristics of the efficiency of query expansion, to what extent increases the probability of a trapped docuscle to be relevant. In terms of probabilities it is formulated as the condition 
\begin{equation}\label{ecboost}
P({\rlv}\vert \pra) > P({\rlv})
\quad\mbox{ or }\quad
\prx>r
\end{equation}
To decide, which query expansion is to be applied, we possess with the following measured frequencies, which we interpret as probabilities:
\begin{equation}\label{emesval}
p=P(\pra \vert \rlv) 
\quad ; \quad
q=P(\pra \vert \nrlv) 
\quad ; \quad
r=P({\rlv})
\end{equation}
and use this or that theoretical model to evaluate $x$. Then $x$ can be obtained from experiments and this can test the applicability of the model.

\section{Implementations}\label{simplem}

\paragraph{Classical setting.} Let $\qs$ be a probability distribution on a sample space $\samps$, and $\pra ,{\rlv}$ being properties, are subsets of $\samps$. 
That means that the appropriate density operators commute, so $\samps$ is a diagonal operator and $\pra ,{\rlv}$ are characteristic functions of appropriate subsets.

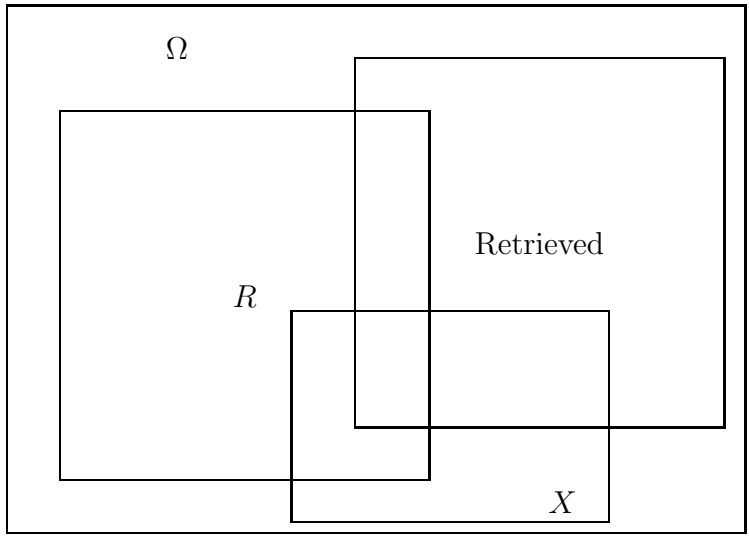
\unitlength.7mm
\[
\begin{picture}(180,100)
\put(0,0){\framebox(140,100)[tr]{}}
\put(30,90){$\samps$}
\put(10,10){\framebox(70,70){\mbox{${\rlv}$}}}
\put(54,2){\framebox(60,40)[rb]{\mbox{$\vphantom{\int} \pra \quad$}}}
\put(66,20){\framebox(70,70){\mbox{Retrieved}}}
\end{picture}
\]
In this case the condition for expanding the query by the term $\pra$ to increase the relevance is formulated in purely classical, set-theoretical terms as follows.

\begin{lemma}\label{tclass}
\begin{equation}\label{eclass1}
P({\rlv}\vert \pra) > P({\rlv})
\quad\mbox{if and only if}\quad
P(\pra \vert {\rlv}) > P(\pra \vert {\nrlv})
\end{equation}
or, in terms of \eqref{emesval}
\begin{equation}\label{eclass2}
\prx > r
\quad\mbox{if and only if}\quad
p > q
\end{equation}
\end{lemma}
\begin{proof}Consider the sequence of equivalent expressions
\[
P({\rlv}\vert \pra) > P({\rlv})
\]
\[
\frac{P({\rlv} \pra)}{P(\pra)} > P({\rlv})
\]
\[
P({\rlv} \pra)>P(\pra)\, P({\rlv})
\]
\[
P(\pra \vert \rlv)\,P(\rlv)>P(\pra)\, P({\rlv})
\]
\[
P(\pra \vert \rlv) >P(\pra)
\]
\[
P(\pra \vert \rlv) >P(\pra \vert {\rlv})\,P(\rlv)+ P(\pra \vert {\nrlv})\,P(\nrlv)
\]
\[
P(\pra \vert {\rlv})\,P(\nrlv)> P(\pra \vert {\nrlv})\,P(\nrlv)
\]
\[
P(\pra \vert {\rlv}) > P(\pra \vert {\nrlv})
\]
This is the reason to call the choice of $\pra $, satisfying $p>q$, natural. 
\end{proof}

\paragraph{Quantum setting.} Suppose we are dealing with quantum systems and the emitted docuscles are in quantum state associated with a density operator $\qs$. Then this measurement is associated with the projection operator ${\rlv}$, and the probability of a docuscle to be relevant is
\[
r=P({\rlv}) = \trc \rlv\qs
\]
The projector associated with the property $\pra$ is denoted also by $\pra$, and we have the following expressions for appropriate conditional probabilities
\begin{equation}\label{emesvalq}
p=P(\pra \vert{{\rlv}})=
\frac{\trc(\rlv \qs\rlv \pra)}{r}
\quad ; \quad
q=P(\pra \vert{{\nrlv}})=
\frac{\trc(\nrlv \qs\nrlv \pra)}{1-r}
\end{equation}
Starting from this, we make query expansion. In terms of Melucci metaphor this looks as follows. Initially, the docuscles are emitted in state $\qs$. After testing the property associated with the projector $\pra $, the beam is in state 
\[
\qs_\pra =
\frac{\pra \qs\pra }{\trc \pra \qs}
\]
Then the relevance test is carried out, and the appropriate probability is:
\begin{equation}\label{epra}
\prx=P({\rlv}\vert \pra)=
\trc \qs_\pra {\rlv}=
\frac{\trc (\pra \qs\pra {\rlv})}{\trc \pra \qs}
\end{equation}
So, the relevance boost is desribed as $P({\rlv}\vert \pra)>P({\rlv})$, and in terms of density operators takes the form:
\begin{equation}\label{ecqboost}
\trc\,(\!\pra \qs\pra {\rlv})>\trc \pra \qs\cdot \trc \qs{\rlv}
\end{equation}
Transform the rhs of the above expression using $\rlv+\nrlv=\1$:
\[
\trc \pra \qs\cdot \trc \qs{\rlv}
=
\left(\trc (\pra \qs\pra \rlv)+\trc (\pra \qs\pra \nrlv)\right)\cdot \trc \qs{\rlv}
\]
therefore the expression \eqref{ecqboost} takes the form
\[
\trc\,(\!\pra \qs\pra {\rlv})(1- \trc \qs{\rlv})>\trc (\pra \qs\pra \nrlv) \cdot \trc \qs{\rlv}
\]
which is equivalent to
\[
\frac{\trc\,(\!\pra \qs\pra {\rlv})}{ \trc \qs{\rlv}}
\,>\,
\frac{\trc (\pra \qs\pra \nrlv)}{ \trc \qs{\nrlv}}
\]
So, like in classical case \eqref{eclass2}, where the choice of a term for expansion is based on Bayesian belief revision (that is, the Law of Total Probability), we have the same criterion, at least on a qualitative level, loosely speaking `making better in quantum style is making better in classical fashion':
\[
\prx > r
\quad\mbox{if and only if}\quad
p > q
\]

\section{Conclusions}\label{scompare}

The starting point for this research were the experimental results obtained by Melucci \cite{Melucci2010}. Their idea was the following. Instead of automated expansion of queries (and then testing the relevance of the retrieved results), the queries were expanded manually, by in certain sense `wrong' terms: they do not obey the Law of Total Probability. The straightforward idea was to model this by quantum mechanics. The evaluations presented in this paper show, however, that when we model the IR process by either classical or quantum experiment, we see that still Bayesian belief revision rules. 

This may have a rather controversial interpretation: Information Retrieval is to be treated as a novel physical reality, whose physical laws are different from both classical and conventional quantum mechanics. Furthermore, physical theories, which were falsified by experiments in `real life', acquire a new chance for rebirth. Yet this needs further investigation.

Throughout the paper the terms occurring were treated as properties of the retrieved documents. Suppose we are now dealing with other kind of resources, which may be not textual. Perhaps the development of methods based on Melucci metaphor might be more appropriate for image or multimedia search.

\end{document}